\documentclass[11pt]{article}
\usepackage{amssymb}
\newtheorem{precor}{{\bf Corollary}}

\newtheorem{precon}{{\bf Conjecture}}

\newtheorem{predefin}{{\bf Definition}}

\newenvironment{defin}[1]{\begin{predefin}{\hspace{-0.5
                   em}{\bf.\ }}{\rm
#1}\hfill{$\spadesuit$}}{\end{predefin}}
\newtheorem{preexm}{{\bf Example}}

\newtheorem{preappl}{{\bf Application}}

\newtheorem{prelem}{{\bf Lemma}}

\newenvironment{lem}{\begin{prelem}{\hspace{-0.5
               em}{\bf.\ }}}{\end{prelem}}
\newtheorem{preproof}{{\bf Proof.\ }}

\newenvironment{proof}[1]{\begin{preproof}{\rm
               #1}\hfill{$\blacksquare$}}{\end{preproof}}
\newtheorem{presproof}{{\bf Sketch of Proof.\ }}

\newtheorem{prethm}{{\bf Theorem}}

\newenvironment{thm}{\begin{prethm}{\hspace{-0.5
               em}{\bf.\ }}}{\end{prethm}}
\newtheorem{prealphthm}{{\bf Theorem}}

\newenvironment{alphthm}{\begin{prealphthm}{\hspace{-0.5
               em}{\bf.\ }}}{\end{prealphthm}}
\newtheorem{prepro}{{\bf Proposition}}

\newtheorem{preprb}{{\bf Problem}}

\def\conct[#1,#2]{\mbox {${#1} \leftrightarrow {#2}$}}
\def\dconct[#1,#2]{\mbox {${#1} \rightarrow {#2}$}}
\def\deg[#1,#2]{\mbox {$d_{_{#1}}(#2)$}}
\def\mindeg[#1]{\mbox {$\delta_{_{#1}}$}}
\def\maxdeg[#1]{\mbox {$\Delta_{_{#1}}$}}
\def\outdeg[#1,#2]{\mbox {$d_{_{#1}}^{^+}(#2)$}}
\def\minoutdeg[#1]{\mbox {$\delta_{_{#1}}^{^+}$}}
\def\maxoutdeg[#1]{\mbox {$\Delta_{_{#1}}^{^+}$}}
\def\indeg[#1,#2]{\mbox {$d_{_{#1}}^{^-}(#2)$}}
\def\minindeg[#1]{\mbox {$\delta_{_{#1}}^{^-}$}}
\def\maxindeg[#1]{\mbox {$\Delta_{_{#1}}^{^-}$}}
\def\isdef{\mbox {$\ \stackrel{\rm def}{=} \ $}}
\def\dre[#1,#2,#3]{\mbox {${\cal E}_{_{#3}}(#1,#2)$}}
\def\pdre[#1,#2,#3]{\mbox {${\cal P}_{_{#3}}(#1,#2)$}}
\def\var[#1,#2]{\mbox {${\rm Var}_{_{#1}}(#2)$}}
\def\ls[#1]{\mbox {$\xi^{^{#1}}$}}
\def\hom[#1,#2]{\mbox {${\rm Hom}({#1},{#2})$}}
\def\onvhom[#1,#2]{\mbox {${\rm Hom^{v}}(#1,#2)$}}
\def\onehom[#1,#2]{\mbox {${\rm Hom^{e}}(#1,#2)$}}
\def\core[#1]{\mbox {$#1^{^{\bullet}}$}}
\def\cay[#1,#2]{\mbox {${\rm Cay}({#1},{#2})$}}
\def\cays[#1,#2]{\mbox {${\rm Cay_{s}}({#1},{#2})$}}
\def\dirc[#1]{\mbox {$\stackrel{\rightarrow}{C}_{_{#1}}$}}
\def\cycl[#1]{\mbox {${\bf Z}_{_{#1}}$}}

\begin{document}
\footnotetext[1]{Correspondence should be addressed to {\tt
hhaji@sbu.ac.ir}.}

\footnotetext[2]{$^\ast$ This paper is partially supported by
Shahid Beheshti University.}
\begin{center}
{\Large \bf Bounds for Visual
Cryptography Schemes}\\
\vspace*{0.5cm}
{\bf Hossein Hajiabolhassan $^\ast$ and Abbas Cheraghi}\\
{\it Department of Mathematical Sciences}\\
{\it Shahid Beheshti University, G.C.}\\
{\it P.O. Box {\rm 1983963113}, Tehran, Iran}\\
{\tt hhaji@sbu.ac.ir}\\
{\it Department of Mathematics} \\
{\it Faculty of Khansar} \\
{\it University of Isfahan}\\
{\it Isfahan, Iran}\\
{\tt cheraghi@sci.ui.ac.ir}\\ \ \\
\end{center}
\begin{abstract}
In this paper, we investigate the best pixel expansion of
the various models of visual cryptography schemes. In this regard,
we consider visual cryptography schemes introduced by Tzeng and
Hu \cite{TZH}. In such a model, only minimal qualified sets can
recover the secret image and that the recovered secret image can
be darker or lighter than the background. Blundo et al.
\cite{BLC} introduced a lower bound for the best pixel expansion
of this scheme in terms of minimal qualified sets. We present
another lower bound for the best pixel expansion of the scheme.
As a corollary, we introduce a lower bound, based on an induced
matching of hypergraph of qualified sets, for the best pixel
expansion of the aforementioned model and the traditional model of
visual cryptography realized by basis matrices. Finally, we study
access structures based on graphs and we present an upper bound
for the smallest pixel expansion in terms of strong chromatic
index.
\begin{itemize}
\item[]{{\footnotesize {\bf Key words:}\ visual cryptography, secret sharing scheme, hypergraph, basis matrices, pixel expansion, contrast.}}
\item[]{ {\footnotesize {\bf Subject classification:} 94A62.}}
\end{itemize}
\end{abstract}
\section{Introduction}
{\it Visual cryptography schemes} (VCS) are a special kind of
secret sharing schemes in which secret is an image. For a set
${\cal P}$ of $n$ participants, a VCS  encrypts a secret image
into $n$  transparencies which constitute the shares given to the
$n$ participants. The power set of participants is usually
divided into {\it qualified} sets, which can visually recover the
secret image by stacking their transparencies without any
cryptography knowledge, and {\it forbidden} sets which have no
information on the secret image.

The fascinating idea of visual cryptography was first introduced by
Naor and Shamir \cite{nash}. Naor and Shamir \cite{nash} have
proved that the pixel expansion of any visual $k$ out of $k$
scheme must be at least $2^{k-1}$. Also, they have presented a
visual $k$ out of $k$ scheme with pixel expansion $2^{k-1}$.

Most papers on visual cryptography investigate two parameters, the
pixel expansion and the contrast. The pixel expansion is the
number of subpixels used to encode each pixel of the secret image
in a share, that should be as small as possible. The contrast
measures the ``difference" between a black and a white pixel in
the reconstructed image. Several results on the contrast and the
pixel expansion of VCSs can be found in \cite{1, 3, 2, BLC, DRO,
4, nash, 6}. Finding the best pixel expansion in the different models
of VCS is the main challenge in visual cryptography. The problem
of determining the best visual contrast (regardless of pixel
expansion) is completely resolved \cite{KU, KUS}, so that, for the
sake of completeness, it is interesting to find bound for the
pixel expansion.

In this paper, we investigate the best pixel expansion of
the different models of visual cryptography schemes. In the second
section, we present several models of visual cryptography
schemes. In the third section, we introduce some lower bounds for
the best pixel expansion of different models of visual
cryptography schemes. In this regard, we consider visual
cryptography schemes introduced by Tzeng and Hu \cite{TZH}. In
such a model, only minimal qualified sets can recover the secret
image and that the recovered secret image can be darker or
lighter than the background. Blundo et al. \cite{BLC} introduced
a lower bound for the best pixel expansion of this scheme in
terms of minimal qualified sets. We present another lower bound
for the best pixel expansion of the scheme. As a corollary, we
introduce a lower bound, based on an induced matching of
hypergraph of qualified sets, for the best pixel expansion of the
aforementioned model and the traditional model of visual
cryptography realized by basis matrices. Finally, we study access
structures based on graphs and we present an upper bound for the
smallest pixel expansion in terms of strong chromatic index.
\section{Models and Notations}
First, we mention some of definitions and notations which are
referred to throughout the paper. Hereafter, the symbol ${\cal P}$
stands for the set of participants. Furthermore, we assume that
${\cal P}=\{1,2,\ldots,n\}$ and let $2^{\cal P}$ denote the power
set of ${\cal P}$, i.e., the set of all subsets of ${\cal P}$. A
family ${\cal Q}\subseteq 2^{\cal P}$ is said to be {\it
monotone} if for any $A\in {\cal Q}$ and any $B\subseteq {\cal
P}$ such that $A\subseteq B$, it holds that $B\in {\cal Q}$.
Throughout the paper we assume that ${\cal F}, {\cal Q}\subseteq
2^{\cal P}$, where ${\cal Q}\cap {\cal F}=\emptyset$, ${\cal
Q}\cup {\cal F}=2^{\cal P}$ and ${\cal Q}$ is monotone. The
members of ${\cal Q}$ and ${\cal F}$ are termed {\it qualified}
sets and {\it forbidden} sets, respectively. Denote the minimal qualified sets of  ${\cal Q}$ by ${\cal Q}_0$. Also, we call
$\Gamma=({\cal P},{\cal
Q},{\cal F})$ the {\it access structure} of the scheme.\\

Let $M=M_{n\times m}$ be an $n\times m$ Boolean matrix. The $i$th
row vector of M is denoted by $M_{i}$. Set $M_{i}\circ M_{j}$ to be
the bit-wise ``OR'' of vectors $M_{i}$ and $M_{j}$. Suppose
$X=\{i_{1},i_{2},\ldots ,i_{q}\}\subseteq \{1, 2, \ldots, n\}$,
and define $M_{_{X}}\isdef M_{i_{1}}\circ M_{i_{2}}\circ
\cdots\circ M_{i_{q}}$, whereas $M[X]\isdef M[X][\{1,\ldots,
m\}]$ denotes the $|X|\times m$ matrix obtained from $M$ by
considering only the rows corresponding to members of $X$. Let
$A||B$ denote the concatenation of two matrices $A$ and $B$ of
the same number of rows. Denote the Hamming weight of row vector
$v$ by $w(v)$. For two vectors $u$ and $v$, denote their inner
product by $u\cdot v$. Let $T$ be a set of vectors. The vector
space generated by $T$ is denoted by $span(T)$.\\
In visual cryptography schemes we assume that the secret image
consists of a collection of black and white pixels. Each pixel of
secret image appears in $n$ modified versions called shares, one
for each transparency, and each share is divided into $m$ black
and white subpixels. The $m$ subpixels of shares can be
represented by an $n\times m$ Boolean matrix $M=[m_{ij}]$ where
$m_{ij}=1$ if and only if $j$th subpixel in the $i$th
transparency is black. The resultant shares should meet the
properties of visual cryptography. The conventional definition
for VCS is as follows.
\begin{defin} {Let $\Gamma=({\cal P},{\cal Q},{\cal F})$ be an access structure where $|{\cal P}|=n$. Two collections
(multisets) ${\cal C}^{0}$ and ${\cal C}^{1}$ of $n\times m$ Boolean matrices
constitute a $(\Gamma ,m)$-${\rm VCS}_1$ if there exist a value
$\alpha(m)>0$ and a set $\{(X, t_{X})\}_{X\in {\cal Q}}$
satisfying
\begin{enumerate}
\item Any qualified set $X=\{i_1,i_2,\ldots ,i_q\}\in {\cal Q}$ can
recover the shared image by stacking their transparencies.
Formally, for any $M\in {\cal C}^{0},\ \omega(M_X)\leq t_{X}-\alpha
(m)\cdot m$, whereas for any $M'\in {\cal C}^{1}$, $\omega (M'_X)\geq
t_X$. \item Any forbidden set $X=\{i_1,i_2,... ,i_q\}\in {\cal
F}$ has no information on the shared image. Formally, the two
collections
 ${\cal D}^t ,t\in \{0,1\}$, of $q\times m$ matrices obtained by restricting
 each $n\times m$ matrix in $M\in {\cal C}^t$ to
rows $i_1, i_2,\ldots , i_q$, are indistinguishable in the sense
that they contain the same matrices with the same frequencies.
\end{enumerate}
The value $m$ is called {\it pixel expansion}, and the value
$\alpha(m)$ is termed {\it contrast}. The first and second
conditions are called {\it contrast} and {\it security},
respectively. The notation $m_1(\Gamma)$ stands for the minimum
value of $m$ for which such a $\Gamma$-${\rm VCS}_1$ exists and
called {\it the best pixel expansion} of $\Gamma$-${\rm VCS}_1$.}
\end{defin}
The most of constructions in this paper are based on two $n\times
m$ matrices, $S^0$ and $S^1$ called {\it basis matrices}. In this case,
the collections ${\cal C}^0$ and ${\cal C}^1$ are generated by permuting the
columns of the corresponding basis matrices $S^0$ and $S^1$,
respectively, in all possible ways.
\begin{defin}{\label{VIS} Let $\Gamma=({\cal P},{\cal Q},{\cal F})$ be an access structure
where $|{\cal P}|=n$. A $(\Gamma
,m)$-${\rm VCS}_2$ is realized using two basis matrices $S^0$ and
$S^1$ of $n\times m$ Boolean matrices if there exist a value
$\alpha(m)>0$ and a set $\{(X, t_{X})\}_{X\in {\cal Q}}$
satisfying
\begin{enumerate}
\item Any qualified set $X=\{i_1,i_2,\ldots ,i_q\}\in {\cal Q}$ can
recover the shared image by stacking their transparencies.
Formally, $\omega(S^0_X)\leq t_{X}-\alpha (m)\cdot m$, whereas
$\omega (S^1_X)\geq t_X$.
\item Any forbidden set
$X=\{i_1,i_2,... ,i_q\}\in {\cal F}$ has no information on the
shared image. Formally, the two $q\times m$ matrices obtained by
restricting $S^0$ and $S^1$ to rows $i_1,i_2,\ldots ,i_q $ are
equal up to a column permutation.
\end{enumerate}

The notation $m_2(\Gamma)$ stands for the minimum value of $m$ for
which such a $\Gamma$-${\rm VCS}_2$ with basis matrices exists and
called {\it the best pixel expansion} of $\Gamma$-${\rm VCS}_2$.}
\end{defin}

Now, we recall the definition of visual cryptography scheme
defined in \cite{TZH}. In this scheme, only the sets in ${\cal
Q}_0$ can recover the secret image by stacking their
transparencies and that the image, which is revealed by stacking
the transparencies of a minimal qualified set, can be darker or
lighter than the background. Note that any non-minimal qualified
set, by stacking their transparencies, has no information on the
shared image, i.e., cannot distinguish a white pixel from a black
one. Here is the formal definition \cite{TZH}.

\begin{defin}{\label{DAR} Let $\Gamma=({\cal P},{\cal Q},{\cal F})$ be an access structure where $|{\cal P}|=n$. A $(\Gamma
,m)$-${\rm VCS}_3$ is realized using two basis matrices $S^0$ and
$S^1$ of $n\times m$ Boolean matrices if there exist a value
$\alpha(m)>0$ and a set $\{(X, t_{X})\}_{X\in {\cal Q}_0}$
satisfying
\begin{enumerate}
\item Any minimal qualified set $X=\{i_1,i_2,\ldots ,i_q\}\in {\cal Q}_0$ can
recover the shared image by stacking their transparencies.
Formally, $\omega(S^0_X)= t_{X}$, whereas either, $\omega
(S^1_X)\geq t_X+\alpha (m)\cdot m$ or, $\omega (S^1_X)\leq
t_X-\alpha (m)\cdot m$.

\item Any forbidden set $X=\{i_1,i_2,\ldots
,i_q\}\in {\cal F}$ has no information on the shared image.
Formally, the two $q\times m$ matrices obtained by restricting
$S^0$ and $S^1$ to rows $i_1,i_2,\ldots ,i_q $ are equal up to a
column permutation.

\item Any non-minimal qualified set $X=\{i_1,i_2,\ldots ,i_q\}\in {\cal Q}\setminus
{\cal Q}_0$, by stacking  their transparencies, has no
information on shared image. Formally, the two $1\times m$
vectors $V_0$ and $V_1$, obtained by OR-ing the rows
$i_1,i_2,\ldots ,i_q $ of the matrix $S^0$ and $S^1$,
respectively, are indistinguishable in the sense that they have
the same Hamming weight.
\end{enumerate}

Also, we will use the notation $m_3(\Gamma)$ to denote the
minimum pixel expansion of basis matrices of $\Gamma$-${\rm VCS}_3$ and
called {\it the best pixel expansion of ${\rm VCS}_3$}.}
\end{defin}

To achieve the smallest pixel expansion in the different models of
visual cryptography schemes, we consider the following definition
in which the minimal qualified subsets can recover the shared
image by stacking their transparencies. Also, the revealed image
can be darker or lighter than the background as well. Moreover,
we don't mind whether non-minimal qualified subsets can obtain
the secret.

\begin{defin} {Let $\Gamma=({\cal P},{\cal Q},{\cal F})$ be an access structure where $|{\cal P}|=n$. Two collections
(multisets) ${\cal C}^{0}$ and ${\cal C}^{1}$ of $n\times m$ Boolean matrices
constitute a $(\Gamma ,m)$-${\rm VCS}_4$ if there exist a value
$\alpha(m)>0$ and a set $\{(X, t_{X})\}_{X\in {\cal Q}_0}$
satisfying
\begin{enumerate}
\item Any minimal qualified set $X=\{i_1,i_2,\ldots ,i_q\}\in {\cal Q}_0$ can
recover the shared image by stacking their transparencies.
Formally, for any $M\in {\cal C}^{0}$, $\omega(M_X)= t_{X}$, whereas for
any $M' \in {\cal C}^{1}$, either, $\omega (M'_X)\geq t_X+\alpha
(m)\cdot m$ or, $\omega (M'_X)\leq t_X-\alpha (m)\cdot m$.

\item Any forbidden set $X=\{i_1,i_2,... ,i_q\}\in {\cal
F}$ has no information on the shared image. Formally, the two
collections  ${\cal D}^t ,t\in \{0,1\}$, of $q\times m$ matrices
obtained by restricting  each $n\times m$ matrix in $M\in {\cal C}^t$ to
rows $i_1, i_2,\ldots , i_q$, are indistinguishable in the sense
that they contain the same matrices with the same frequencies.
\end{enumerate}

The notation $m_4(\Gamma)$ stands for the minimum value of $m$ for
which such a $\Gamma$-${\rm VCS}_4$ exists and called {\it the
best pixel expansion} of $\Gamma$-${\rm VCS}_4$.}
\end{defin}

Most constructions in this paper are realized using basis
matrices; hence, we consider the following definition as well.

\begin{defin}{\label{DARNEW} Let $\Gamma=({\cal P},{\cal Q},{\cal F})$ be an access structure where $|{\cal P}|=n$.
A $(\Gamma, m)$-${\rm VCS}_5$ is realized using two basis matrices
$S^0$ and $S^1$ of $n\times m$ Boolean matrices if there exist a
value $\alpha(m)>0$ and a set $\{(X, t_{X})\}_{X\in {\cal Q}_0}$
satisfying
\begin{enumerate}
\item Any minimal qualified set $X=\{i_1,i_2,\ldots ,i_q\}\in {\cal Q}_0$ can
recover the shared image by stacking their transparencies.
Formally, $\omega(S^0_X)= t_{X}$, whereas either, $\omega
(S^1_X)\geq t_X+\alpha (m)\cdot m$ or, $\omega (S^1_X)\leq
t_X-\alpha (m)\cdot m$.

\item Any forbidden set $X=\{i_1,i_2,\ldots
,i_q\}\in {\cal F}$ has no information on the shared image.
Formally, the two $q\times m$ matrices obtained by restricting
$S^0$ and $S^1$ to rows $i_1,i_2,\ldots ,i_q $ are equal up to a
column permutation.
\end{enumerate}

Also, we will use the notation $m_5(\Gamma)$ to denote the
minimum pixel expansion of basis matrices of $\Gamma$-${\rm VCS}_5$ and
called {\it the best pixel expansion of ${\rm VCS}_5$}.}
\end{defin}

It is instructive to add some notes on different models we have
introduced so far. First, ${\rm VCS}_1$ has been considered. In
fact, ${\rm VCS}_1$ is the traditional model of VCS which was
introduced by M. Naor and A. Shamir \cite{nash}. Constructing the
families ${\cal C}^0$ and ${\cal C}^1$, mentioned in ${\rm VCS}_1$, may seem a
daunting task. However, it can be more convenient to handle with
basis matrices. Hence, basis matrices are used in the most
constructions of VCS found in the literature. That is why we
consider the models ${\rm VCS}_2$, ${\rm VCS}_3$ and ${\rm
VCS}_5$. Finally, we consider ${\rm VCS}_4$ to achieve the
smallest pixel expansion among the different models of visual
cryptography schemes.
\section{Lower Bounds for Pixel Expansion}
In this section, we introduce some lower bounds for the best pixel
expansion of the different models of visual cryptography schemes.
First, for a given access structure $\Gamma =({\cal P}, {\cal Q},
{\cal F})$, we present a lower bound for the best pixel expansion
of $\Gamma $-${\rm VCS}_2$ as follows.

\begin{thm}\label{main}
Let $\Gamma =({\cal P}, {\cal Q}, {\cal F})$ be an access
structure and  let $\Omega=\{F_1, F_2,\ldots , F_t\}$ be a
collection of forbidden sets such that $\displaystyle{\bigcup
_{i=1}^t}F_i \in {\cal F}$. Also, assume that for any two
disjoint non-empty subsets $A, B \subset \Omega$, there exists a
forbidden set $F'\in {\cal F}$ such that at least one of the
following conditions holds

\begin{itemize}
\item For any $F_i \in A$, $F_i\cup F' \in {\cal F}$ and there
exists an $F_j\in B$ such that $F_j\cup F' \in {\cal Q}$.

\item For any $F_i \in B$, $F_i\cup F' \in {\cal F}$ and there
exists an $F_j\in A$ such that $F_j\cup F' \in {\cal Q}$.
\end{itemize}

Then we have $m_2(\Gamma) \geq t+1$.
\end{thm}
\begin{proof}
{Let  $S^0$ and $S^1$ be $n\times m_2(\Gamma)$ the basis matrices of
access structure $\Gamma =({\cal P}, {\cal Q}, {\cal F})$.
Consider the following sets
$$T^0\isdef \{S^0_{_{F_i}}\ | \ 1\leq i \leq t\},$$
$$T^1\isdef \{S^1_{_{F_i}}\ |\ 1\leq i \leq t\}.$$
We claim that the vectors of $T^0$ (resp. $T^1$) are linearly
independent over the real numbers. On the contrary, suppose that
there are some real coefficients $a_{_{F_i}}$ such that
$\displaystyle{\sum_{i=1}^t} a_{_{F_i}} S^0_{_{F_i}}=0$. Define
$Y=\displaystyle{\bigcup _{i=1}^t}F_i$. Since $Y \in {\cal F}$,
$S^0[Y]$ and $S^1[Y]$ are equal up to a permutation of columns;
hence, $\displaystyle{\sum_{i=1}^t} a_{_{F_i}} S^1_{_{F_i}}= 0$.
Therefore,

$$\displaystyle{\sum_{a_{_{F_i}} > 0}} a_{_{F_i}}
S^0_{_{F_i}}=- \displaystyle{\sum_{a_{_{F_j}} < 0}} a_{_{F_j}}
S^0_{_{F_j}}\quad \& \quad \displaystyle{\sum_{a_{_{F_i}}
> 0}} a_{_{F_i}} S^1_{_{F_i}}=-\displaystyle{\sum_{a_{_{F_j}} < 0}}
a_{_{F_j}} S^1_{_{F_j}}.
$$
Set $A\isdef \{F_j| a_{_{F_j}} < 0\}$ and $B\isdef \{F_i|
a_{_{F_i}} > 0\}$. Without loss of generality, assume that there
exists a forbidden set $F'$ such that for any $F_j \in A$,
$F_j\cup F' \in {\cal F}$ and there exists a member of $B$ such
as $F''$ such that $F''\cup F' \in {\cal Q}$. We have
\begin{equation}\label{coff}
S^0_{_{F'}}\cdot\displaystyle{\sum_{a_{_{F_i}} > 0}} a_{_{F_i}}
S^0_{_{F_i}}= S^0_{_{F'}}\cdot \displaystyle{\sum_{a_{_{F_j}} <
0}} -a_{_{F_j}} S^0_{_{F_j}}\ \ \& \ \ S^1_{_{F'}}\cdot
\displaystyle{\sum_{a_{_{F_i}}
> 0}} a_{_{F_i}} S^1_{_{F_i}}= S^1_{_{F'}}\cdot \displaystyle{\sum_{a_{_{F_j}} < 0}}
-a_{_{F_j}} S^1_{_{F_j}}.
\end{equation}
Note that $w(S^0_{_{F'}})=w(S^1_{_{F'}})$. Also, for any $F_j \in
A$, $F'\cup F_j  \in {\cal F}$; consequently, $S^0[F_j\cup F']$
and $S^1[F_j\cup F']$ are equal up to a permutation of columns.
Hence, $S^0_{_{F'}}\cdot S^0_{_{F_j}} =S^1_{_{F'}}\cdot
S^1_{_{F_j}}$. Moreover, there exists an $F'' \in B$ such that
$F''\cup F' \in {\cal Q}$. Thus,
$$S^0_{_{F'}}\cdot
S^0_{_{F''}}  > S^1_{_{F'}}\cdot S^1_{_{F''}}.$$
In view of Equations~~\ref{coff}, one can obtain that for any
$F_i \in \Omega$, $a_{_{F_i}}=0$; accordingly,  $\dim({\rm
span}(T^0))= \dim({\rm span}(T^1))= t$. In addition, $F'\cup Y\in
{\cal Q}$; therefore, $w(S^0_{_{F'\cup Y}}) \not = w(S^1_{_{F'\cup
Y}})$. Also,  $w(S^0_{_{F'}})=w(S^1_{_{F'}})$ and the matrices
$S^0[Y]$ and $S^1[Y]$ are equal up to a permutation of columns.
Now, it is easy to check that $\dim({\rm span}(T^1\cup
\{S^1_{_{F'}}\}))=t+1$. On the other hand, $m_2(\Gamma) \geq
\dim({\rm span}(T^1\cup \{S^1_{_{F'}}\}))$. Thus, $m_2(\Gamma)
\geq t+1$.}
\end{proof}
For a given access structure $\Gamma=({\cal P}, {\cal Q}, {\cal
F})$, a lower bound for the best pixel expansion of $\Gamma
$-${\rm VCS}_3$ has been introduced in \cite{BLC} as follows.

\begin{alphthm}{\rm \cite{BLC}} \label{BBLL}
Let $\Gamma=({\cal P}, {\cal Q}, {\cal F})$ be an access
structure. The best pixel expansion of $\Gamma $-${\rm VCS}_3$
satisfies
$$m_3(\Gamma) \geq \lceil{|{\cal Q}_0| \over 2}\rceil.$$
\end{alphthm}

Note that the aforementioned theorem is not effective when
$|{\cal Q}_0|$ is small. Now, we present a theorem which can be
considered as a counterpart of Theorem \ref{BBLL}.

\begin{thm}\label{mainGEN}
Let $\Gamma =({\cal P}, {\cal Q}, {\cal F})$ be an access
structure and let $\Omega=\{F_1, F_2,\ldots , F_t\}$ be a
collection of forbidden sets such that $\displaystyle{\bigcup
_{i=1}^t}F_i \in {\cal F}$. Also, assume that for any non-empty
subset $A\subset \Omega$, there exist two forbidden sets $F'\in
A$ and $F''\in {\cal F}$ such that $F'\cup F'' \in {\cal Q}_0$
and for any $F_i \in A\setminus \{F'\}$, $F_i\cup F''\not \in
{\cal Q}_0$. Then $m_3(\Gamma) \geq t+1$.
\end{thm}
\begin{proof}{
It is simple to prove that $m_3(\Gamma) \geq 2$ whenever $t=1$;
hence, assume that $t\geq 2$.
Let  $S^0$ and $S^1$ be $n\times m_3(\Gamma)$ the basis matrices of
access structure $\Gamma =({\cal P}, {\cal Q}, {\cal F})$. Set
$F_{t+1} \isdef F_1\cup \cdots\cup F_t$. Consider the following
sets
$$T^0\isdef \{S^0_{_{F_i}}\ | \ 1\leq i \leq t+1\},$$
$$T^1\isdef \{S^1_{_{F_i}}\ |\ 1\leq i \leq t+1\}.$$
We claim that the vectors of $T^0$ (resp. $T^1$) are linearly
independent over the real numbers. Suppose that there are some
real coefficients $a_{_{F_i}}$ such that
$\displaystyle{\sum_{i=1}^{t+1}} a_{_{F_i}} S^0_{_{F_i}}=0$.
Since $F_{t+1} \in {\cal
F}$, $S^0[F_{t+1}]$ and $S^1[F_{t+1}]$ are equal up to a permutation of
columns; consequently, $\displaystyle{\sum_{i=1}^{t+1}} a_{_{F_i}}
S^1_{_{F_i}}= 0$. Set $A\isdef \{F_i|\ i\leq t,\ a_{_{F_i}} \not =
0\}$. If $A=\emptyset$, then the assertion follows easily. Hence,
assume that $A\not=\emptyset$ and there exist two forbidden sets
$F'\in A$ and $F''\in {\cal F}$ such that $F'\cup F'' \in {\cal
Q}_0$ and for any $F_i \in A\setminus \{F'\}$, $F_i\cup F''\not
\in {\cal Q}_0$. We have
\begin{equation}\label{coffnew}
S^0_{_{F''}}\cdot\displaystyle{\sum_{a_{_{F_i}} \not = 0}}
a_{_{F_i}} S^0_{_{F_i}}=0\ \ \quad \& \ \ \quad  S^1_{_{F''}}\cdot
\displaystyle{\sum_{a_{_{F_i}} \not = 0}} a_{_{F_i}}
S^1_{_{F_i}}=0.
\end{equation}
Note that $w(S^0_{_{F''}})=w(S^1_{_{F''}})$. Also, for any $F_i
\in A\cup\{F_{t+1}\}\setminus \{F'\}$, $F''\cup F_i \not \in
{\cal Q}_0$; consequently, $S^0_{_{F''}}\cdot S^0_{_{F_i}}
=S^1_{_{F''}}\cdot S^1_{_{F_i}}$. Moreover, $F''\cup F' \in {\cal
Q}_0$. Thus,
$$S^0_{_{F''}}\cdot
S^0_{_{F'}}  \not = S^1_{_{F''}}\cdot S^1_{_{F'}}.$$
In view of Equations~~\ref{coffnew}, one can obtain that for any
$F_i \in A$, $a_{_{F_i}}=0$; accordingly,  $\dim({\rm span}(T^0))=
\dim({\rm span}(T^1))= t+1$ which implies that $m_3(\Gamma) \geq
t+1$.
}\end{proof}
In the language of hypergraph theory, for a given access
structure $\Gamma=({\cal P}, {\cal Q}, {\cal F})$, one can
introduce a lower bound for the best pixel expansion of $\Gamma
$-${\rm VCS}_2$ and $\Gamma $-${\rm VCS}_3$ in terms of an induced
matching of the hypergraph $({\cal P},{\cal Q})$.
\begin{thm}
\label{genk} Let $\Gamma=({\cal P}, {\cal Q}, {\cal F})$ be an
access structure. Also, assume that there exist disjoint
qualified sets $A_1,\ldots ,A_t$ such that for any qualified set
$B\subseteq A_1\cup \cdots \cup A_t$ one should have
$A_i\subseteq B$ for some $1\leq i \leq t$. Then
$$\min\{m_2(\Gamma), m_3(\Gamma)\} \geq
\displaystyle \sum_{i=1}^t 2^{|A_i|-1}-(t-1).$$
\end{thm}
\begin{proof}{
Suppose that $S^0$ and $S^1$ are $n\times m_3(\Gamma)$ basis
matrices for access structure $\Gamma =({\cal P}, {\cal Q}, {\cal
F})$. Let $|A_i|=r_i$ and $A_i=\{p_{i1},\ldots ,p_{ir_i}\}$.
Define $Y_i\isdef \{p_{i1},\ldots ,p_{i(r_i-1)}\}$. Consider
the non-empty members of power set of $Y_i$'s. Set $\{F_{i1},\ldots,
F_{i(2^{r_i-1}-1)}\}\isdef 2^{Y_i}\setminus\emptyset$ such that
$|F_{i1}|\leq |F_{i2}|\leq \cdots \leq|F_{i(2^{r_i-1}-1)}|$.
Define
$$\Omega\isdef \{ F_{ij}|\ 1\leq i \leq t,\ 1\leq j \leq 2^{r_{_{i}}-1}-1\}.$$
Consider the following ordering for $\Omega$,
$$F_{11}\leq \cdots \leq F_{1(2^{r_{_{1}}-1}-1)}\leq F_{21}\leq \cdots \leq F_{2(2^{r_{_{2}}-1}-1)}\leq \cdots
\leq F_{t(2^{r_{_{t}}-1}-1)}.$$
Assume that $A$ is a non-empty subsets of $\Omega$. Without loss
of generality, suppose that $F_{rs}$ is the largest member of $A$.
Set $F'\isdef A_r \setminus F_{rs}$. It is straightforward to
check that for any $F_{ij} \in A\setminus \{F_{rs}\}$ we have
$F_{ij}\cup F' \in {\cal F}$, whereas $F_{rs}\cup F' \in {\cal
Q}_0$. In view of Theorem \ref{mainGEN}, one can conclude that
$m_3(\Gamma) \geq \sum_{i=1}^t 2^{|A_i|-1}-(t-1)$. Similarly, one
can show that $m_2(\Gamma) \geq \sum_{i=1}^t 2^{|A_i|-1}-(t-1)$,
as desired.}
\end{proof}

Access structure $\Gamma =({\cal P}, {\cal Q}, {\cal F})$ with
$|{\cal P}|=k$ and ${\cal Q}=\{{\cal P}\}$ is well-known as  $k$
out of $k$ scheme. Theorem \ref{genk} presents a simple proof
that the pixel expansion of basis matrices of $k$ out of $k$
scheme is at least $2^{k-1}$.
\begin{precor}{\rm \cite{nash}}
Let $\Gamma=({\cal P}, {\cal Q}, {\cal F})$ be a $k$ out of $k$
scheme. Then
$$\min \{m_2(\Gamma), m_3(\Gamma)\} \geq 2^{k-1}.$$
\end{precor}

Now, we introduce a lower bound for the best pixel expansion of
$\Gamma $-${\rm VCS}_5$. One can deduce the following theorem
whose proof is almost identical to that of Theorem \ref{mainGEN}
and the proof is omitted for the sake of brevity.

\begin{thm}\label{NEWMOD}
Let $\Gamma =({\cal P}, {\cal Q}, {\cal F})$ be an access
structure and let $\Omega=\{F_1, F_2,\ldots , F_t\}$ be a
collection of forbidden sets such that $\displaystyle{\bigcup
_{i=1}^t}F_i \in {\cal F}$. Also, assume that for any non-empty
subset $A\subset \Omega$, there exist two forbidden sets $F'\in
A$ and $F''\in {\cal F}$ such that $F'\cup F'' \in {\cal Q}_0$
and for any $F_i \in A\setminus \{F'\}$, $F_i\cup F'' \in {\cal
F}$. Then we have $m_5(\Gamma) \geq t$.
\end{thm}
\section{Graph Access Structure}
In this section, we study access structures based on graphs. To
begin, some definitions are given which are used throughout this
section. A graph access structure is an access structure for which
the set of participants is the vertex set $V(G)$ of a graph
$G=(V(G), E(G))$, and the edge set of $G$ constitutes the minimal
qualified subsets of access structure, i.e., the qualified
subsets are precisely those containing an edge of G. From $G$, one
can define an access structure $G=(V(G), {\cal Q}, {\cal F})$
where ${\cal Q}_0=E(G)$.

Throughout the paper the word graph is used for a finite simple
graph. A subgraph $H$ of a graph $G$ is said to be induced if for
any pair of vertices $u$ and $v$ of $H$, $\{u,v\}$ is an edge of
$H$ if and only if $\{u,v\}$ is an edge of $G$. Two graphs $G$ and
$H$ are called disjoint if they have no vertex in common. A {\it
matching} $M_n$ is a set of $n$ disjoint edges, that is, no two
edges share a common vertex.  A subgraph of $G$ whose edge set is
non-empty and forms a complete bipartite graph is called a {\it
biclique} of $G$. A {\it biclique cover} ${\cal B}$ of $G$ is a
collection of bicliques covering $E(G)$ (every edge of G belongs
to at least one biclique of the collection). The {\it biclique
covering number} of $G$, $bc(G)$, is the fewest number of
bicliques among all biclique covers of $G$.

A {\it homomorphism} $f: G \longrightarrow H$ from a graph $G$ to
a graph $H$ is a map $f: V(G) \longrightarrow V(H)$ such that
$\{u,v\} \in E(G)$ implies $\{f(u),f(v)\} \in E(H)$. Notation
$\onehom[G,H]$ denotes the sets of {\it onto--edges}
homomorphisms from $G$ to $H$, for more on graph homomorphism see
\cite{DAHA1, HN}.

\begin{lem}\label{ontohom}
Let $G$ and $H$ be two graphs such that $\onehom[G,H]\not =
\emptyset$. Then $m_1(G)\geq m_1(H)$ and $m_2(G)\geq m_2(H)$.
\end{lem}
\begin{proof}{
First, we show that $m_1(G)\geq m_1(H)$. Without loss of
generality, suppose that $H$ does not have any isolated vertex.
Assume that $V(G)=\{1,\ldots,n\}$, $V(H)=\{1,\ldots,n'\}$, and
$\sigma \in \onehom[G,H]$. Also, let ${\cal C}^{0}$ and ${\cal C}^{1}$ be two
collections (multisets) of $n\times m_1(G)$ Boolean matrices
constitute a $(G ,m_1(G))$-${\rm VCS}_1$. For any $n\times
m_1(G)$ matrix $M$, define $n'\times m_1(G)$ matrix $M_{\sigma}$
as follows. For any $1\leq i \leq n'$, the $i$th row of
$M_{\sigma}$ is the vector $M_{\sigma^{-1}(i)}$; i.e., the vector
obtained by considering the bit-wise ``OR'' of the vectors corresponding to
participants in ${\sigma^{-1}(i)}$. Now, construct two
collections (multisets) of $n'\times m_1(G)$ Boolean matrices
${\cal D}^{0}$ and ${\cal D}^{1}$ as follows.
$${\cal D}^0\isdef \{M_{\sigma}|\ M\in {\cal C}^{0}\}\ \ \quad  \& \ \ \quad {\cal D}^1\isdef \{N_{\sigma}|\ N\in {\cal C}^{1}\}.$$
It is easy to check that ${\cal D}^{0}$ and ${\cal D}^{1}$ constitute an $(H
,m_1(G))$-${\rm VCS}_1$. Similarly, one can show that $m_2(G)\geq
m_2(H)$, as claimed.
}
\end{proof}

Now, we provide an upper bound for $m_4(G)$. First, we specify
the exact value of $m_4(M_n)$ as follows.

\begin{lem}\label{bi}
Let $G$ be a graph such that each connected component of $G$ is a
biclique or an isolated vertex. Then $m_4(G)=2$.
\end{lem}
\begin{proof}{
First, we prove the assertion when $G$ is a matching. Let $M_n$ be
a matching with $n$ edges where $V(M_n)=\{1,2,\ldots,2n\}$ and
$E(M_n)=\{\{2i-1,2i\}|\ 1\leq i \leq n\}$. Set
$${\cal D}^0=\{\left ( {\begin{array}{cc}
  1 & 1 \\
  0 & 0
\end{array}}\right ),\
\left ( {\begin{array}{cc}
  0 & 0 \\
  1 & 1
\end{array}}\right )\}$$
and
$${\cal D}^1=\{\left ( {\begin{array}{cc}
  1 & 0 \\
  0 & 1
\end{array}}\right ),\
\left ( {\begin{array}{cc}
  0 & 1 \\
  1 & 0
\end{array}}\right )\}.$$
For $i=0,1$, define
$${\cal C}^i\isdef \{(A^i_1||\cdots||A^i_n)^T|\ A^i_j\in {\cal D}^i,\ 1\leq j\leq n\},$$
where $(A^i_1||\cdots||A^i_n)^T$ means the transpose of the
matrix $A^i_1||\cdots||A^i_n$. Now, one can check that ${\cal C}^0$ and
${\cal C}^1$ constitute an $(M_n ,2)$-${\rm VCS}_4$; that is,
$m_4(M_n)=2$. Note that adding isolated vertices does not alter
the pixel expansion of $G$. Since if $G$ has isolated vertices,
then one can add a zero row to any matrix of ${\cal C}^0$ and ${\cal C}^1$
corresponding to any isolated vertex. It is readily seen that the
new collections of matrices constitute a $(G ,2)$-${\rm VCS}_4$.
Similarly, if each connected components of $G$ is a biclique or
an isolated vertex, then one can show that $m_4(G)=2$.
}
\end{proof}

A strong edge coloring of a graph $G$ is an edge coloring in
which every color class is an induced matching; that is, any two
vertices belonging to distinct edges with the same color are not
adjacent. The strong chromatic index $s'(G)$ is the minimum
number of colors in a strong edge coloring of $G$. It is
well-known that $s'(G)\leq 2\Delta(G)^2$, see \cite{MR}. Now, we
present an upper bound for $m_4(G)$ in terms of strong chromatic
index and biclique covering number.

\begin{thm}\label{bcs}
Let $G$ be a non-empty graph. Then $m_4(G) \leq \min\{2bc(G),
2s'(G)\}$.
\end{thm}
\begin{proof}{
It is well-known that $m_4(G)\leq m_1(G) \leq 2bc(G)$, see
\cite{1}. Consider a strong edge coloring with $s'(G)$ colors.
Let $I_1,\ldots ,I_{s'(G)}$ be the color classes of the strong
edge coloring. For any $1\leq i \leq s'(G)$, one can extend any
$I_i$ to a spanning subgraph of $G$, say $J_i$, such that $I_i$
is an induced subgraph of $J_i$ and $E(J_i)\setminus
E(I_i)=\emptyset$. In view of Lemma \ref{bi}, for any $1\leq i
\leq s'(G)$, there exist two collections ${\cal C}^0_i$ and ${\cal C}^1_i$ of
matrices which constitute a $(J_i,2)$-${\rm VCS}_4$. For $i=0,1$,
set
$${\cal D}^i\isdef \{A^i_1||\cdots||A^i_{s'(G)}|\ A^i_j\in {\cal C}^i_j,\ 1\leq j\leq s'(G)\}.$$
It is easy to see that ${\cal D}^0$ and ${\cal D}^1$ constitute a
$(G,2s'(G))$-${\rm VCS}_4$.
}
\end{proof}

A {\it $t$-strong biclique covering} of a graph $G$ is an edge
covering, $E(G)=\displaystyle{ \cup_{i=1}^t} E(H_i)$, where each
$H_i$ is a set of disjoint bicliques, say $H_{i1}, \ldots,
H_{ir_i}$, such that the graph $G$ has no edges between $H_{ik}$
and $H_{ij}$ for any $1 \leq j < k \leq r_i$. The {\it strong
biclique covering number} $s(G)$ is the minimum number $t$ for
which there exists a $t$-strong biclique covering of $G$. It is
easy to verify that $s(G)\leq \min\{bc(G), s'(G)\}$. The proof of
the next theorem is identical to that of Theorem \ref{bcs} and
the proof is omitted for the sake of brevity. Here is a
generalization of Theorem \ref{bcs}.

\begin{thm}
Let $G$ be a non-empty graph. Then we have $m_4(G) \leq 2s(G)$.
\end{thm}

Suppose that $P_{_{4}}$ is a path with the vertex set $\{v_{_{1}},
v_{_{2}},v_{_{3}}, v_{_{4}}\}$
 and the edge set $\{\{v_{_{1}},v_{_{2}}\},\{v_{_{2}},v_{_{3}}\},\{v_{_{3}},v_{_{4}}\}\ \}$.
 Set $F{_{1}}\isdef \{v_{_{1}}\}$ and $F{_{2}}\isdef
 \{v_{_{3}}\}$. It is easy to see that $F{{_1}}$ and $F{{_2}}$ satisfy
Theorem \ref{main}; consequently, $m_2(P_{_{4}})\geq 3$.
Furthermore, it is easy to check that
$$S^0=\left ( {\begin{array}{ccc}
0 & 1 & 0\\
0 & 1 & 1\\
0 & 1 & 1\\
0 & 0 & 1
\end{array}}\right )\quad \quad {\rm and}\quad \quad
S^1=\left ( {\begin{array}{ccc}
0 & 1 & 0\\
1 & 0 & 1\\
0 & 1 & 1\\
1 & 0 & 0
\end{array}}\right )$$

are the basis matrices of $(P_{_{4}} ,3)$-${\rm VCS}_2$. Thus, $m_2(P_{_{4}})=3$
which implies that the lower bound mentioned in Theorem
\ref{main} is sharp.

The following corollary is a special case of Theorem \ref{genk}.
\begin{precor}\label{GM}
Let $G$ be a graph access structure and $e_1,\ldots ,e_t$ be an
induced matching of $G$. Then we have
$$\min \{m_2(G), m_3(G)\}\geq t+1.$$
\end{precor}

Now, we show that $m_3(M_2)=3$. Consider the following matrices
$$S^0=\left ( {\begin{array}{ccc}
1 & 0 & 1\\
0 & 1 & 1\\
1 & 1 & 0\\
1 & 1 & 0
\end{array}}\right )\quad \quad {\rm and}\quad \quad
S^1=\left ( {\begin{array}{ccc}
1 & 0 & 1\\
1 & 0 & 1\\
1 & 1 & 0\\
0 & 1 & 1
\end{array}}\right ).$$

One can check that $S^0$ and $S^1$ are the basis matrices of $(M_2
,3)$-${\rm VCS}_3$ which implies that the lower bound mentioned
in Corollary \ref{GM} is sharp.

\ \\
{\bf Acknowledgement:} The authors wish to thank anonymous
referees for their invaluable comments.

\end{document}